\pdfoutput=1

\documentclass{llncs}
\usepackage[ruled, noline, algo2e, noend, linesnumbered, nokwfunc]{algorithm2e}
\usepackage{amsmath,amssymb,cancel}   
\usepackage{longtable}
\usepackage{pgfplotstable}
\usepackage{cite}
\usepackage{url}
\usepackage{calc}
\usepackage{lscape}

\newcommand{\TB}{T_{\text{B}}}
\newcommand{\TSPQR}{T_{\text{SPQR}}}
\newcommand{\CPP}{C\nolinebreak\hspace{-.05em}\raisebox{.6ex}{\tiny\bf++}\nolinebreak\hspace{-.10em}\xspace}
\newcommand{\heinztwo}{Heinz 2.0\xspace}
\newcommand{\ACTMOD}{\texttt{ACTMOD}\xspace}
\newcommand{\JMPALM}{\texttt{JMP\_ALM}\xspace}
\newcommand{\JMP}{\texttt{JMP}\xspace}
\newcommand{\CRR}{\texttt{CRR}\xspace}
\newcommand{\PUCNU}{\texttt{PUCNU}\xspace}
\newcommand{\isixhundertforty}{\texttt{i640}\xspace}
\newcommand{\Hinstances}{\texttt{H}\xspace}
\newcommand{\HTwoinstances}{\texttt{H2}\xspace}
\newcommand{\RANDOM}{\texttt{RANDOM}\xspace}

\SetKwFunction{preprocess}{Preprocess}
\SetKwFunction{merge}{Merge}
\SetKwFunction{remove}{Remove}
\SetKwFunction{isolate}{Isolate}
\SetKwFunction{processbicomponent}{ProcessBicomponent}
\SetKwFunction{processtricomponent}{ProcessTricomponent}
\SetKwFunction{solveunrooted}{SolveUnrooted}
\SetKwFunction{solverooted}{SolveRooted}
\DontPrintSemicolon
\SetArgSty{textrm}
\SetFuncSty{textsc}

\begin{document}
\pagestyle{headings}  
\mainmatter              
\title{Solving the Maximum-Weight Connected Subgraph Problem to Optimality}
\titlerunning{Solving the Maximum-Weight Connected Subgraph Problem to Optimality}  
%
\author{Mohammed El-Kebir\inst{1,2,3} \and Gunnar W.\ Klau\inst{2,3}
}
\authorrunning{Mohammed El-Kebir and Gunnar W.\ Klau} 
%
%
\institute{$^1$~Department of Computer Science and Center for Computational
  Molecular Biology, Brown University, Providence,
  RI, 02906, USA\\
  $^2$ Life Sciences group, Centrum Wiskunde \& Informatica (CWI), Science Park 123, 1098 XG
  Amsterdam, the Netherlands\\
  $^3$ Centre for Integrative Bioinformatics VU (IBIVU), VU University
  Amsterdam, De Boelelaan 1081A, 1081 HV Amsterdam, the Netherlands\\
  \email{melkebir@cs.brown.edu, gunnar.klau@cwi.nl}}

\maketitle              

\begin{abstract}
Given an undirected node-weighted graph, the Maximum-Weight Connected Subgraph
problem (MWCS) is to identify a subset of nodes of maximal
sum of weights that
induce a connected subgraph. MWCS is closely related to the
well-studied Prize-Collecting Steiner Tree
problem and has many applications in different areas, including
computational biology, network design and computer vision. The problem is
NP-hard and even hard to approximate within a
constant factor. In this work we describe an algorithmic scheme for
solving MWCS to provable optimality, which is based on preprocessing
rules, new results on
decomposing an instance into its biconnected and triconnected components and
a branch-and-cut approach combined with a primal heuristic. We demonstrate the
performance of our
method on the benchmark instances of the 11th DIMACS implementation challenge
consisting of MWCS as well as transformed PCST instances.
\end{abstract}

\section{Introduction}
\label{sec:introduction}
We consider the Maximum-Weight Connected Subgraph
problem (MWCS). Giv\-en an undirected node-weighted graph, the task is to
find a subset of nodes of maximal sum of weights that
induce a connected subgraph. A formal definition of the unrooted and rooted
variant is as follows.

\begin{definition}[MWCS]
Given an undirected graph $G = (V, E)$ with node weights $w: V \to
\mathbb{R}$, find a subset $V^* \subseteq V$ such that the induced
graph $G[V^*] := \left(V^*, E \cap {{V^*} \choose 2}\right)$ is connected and the
weight $w(G[V^*]) := \sum\limits_{v \in V^*}w(v)$ is maximal. 
\end{definition}

\begin{definition}[R-MWCS]
Given an undirected graph $G = (V, E)$, a node set $R \subseteq V$ and node weights $w: V \to
\mathbb{R}$, find a subset $V^* \subseteq V$ such that $R \subseteq V^*$, the induced
graph $G[V^*] := \left(V^*, E \cap {{V^*} \choose 2}\right)$ is connected and the
weight $w(G[V^*]) := \sum\limits_{v \in V^*}w(v)$ is maximal. 
\end{definition}


Johnson mentioned MWCS in his NP-completeness column~\cite{Johnson:1985cz}. The
problem and its cardinality-constrained and budget variants have 
numerous important applications in different areas, including designing
fiber-optic networks \cite{Lee:1998uk},
oil-drilling \cite{Hochbaum:1994hj}, systems biology
\cite{Yamamoto:2009bc,Dittrich:2008hf,Backes:2012us}, wildlife
corridor design \cite{Dilkina:2010wh}, computer vision
\cite{ChaoYehChen:2012jh} and forest planning \cite{Carvajal:2013tc}.

The maximum-weight connected subgraph problem is closely related to
the well-studied Prize-Collecting Steiner Tree problem (PCST) \cite{Johnson:2000vs,lwpkmf:pcst:2006}, which is
defined as follows.

\begin{definition}[PCST]
Given an undirected graph $G = (V, E)$ with node profits $p: V \to
\mathbb{R}_{\geq 0}$ and edge costs $c: E \to \mathbb{R}_{\geq 0}$, find a
connected subgraph $T = (V^*, E^*)$ of $G$
such that $p(T) := \sum\limits_{v \in V^*}p(v) -
\sum\limits_{e \in E^*}c(e)$ is maximal.
\end{definition}

In \cite{Dittrich:2008hf} we described a reduction from MWCS to
PCST and showed that a prize-collecting Steiner tree $T$ in the
transformed instance is a connected subgraph in the original instance
with weight $p(T) - w'$, where $w'$ is the minimum weight of a
node. We also gave a simple approximation-preserving reduction from PCST to
MWCS: Given an instance $(G = (V, E), p, c)$ of PCST, the corresponding
instance $(G', w)$ of MWCS is obtained by splitting each edge $(v,w)$ in $E$
into two edges $(v, u)$ and $(u, w)$, and setting the weight $w(u)$ of the
introduced split vertex $u$ to $-c(e)$.

\begin{theorem}\label{thm:1}
  A maximum-weight connected subgraph $T'$ in the transformed instance
  corresponds to an optimal prize-collecting Steiner tree $T$ in the
  original instance, and $w(T') = p(T)$.  
\end{theorem}

\begin{proof}
  We first observe that if a split vertex $u$ is part of $T'$, then
  also its neighbors $v$ and $w$ must be in $T'$, otherwise $T'
  \setminus \{u\}$ would be a better solution. We then can simply map each split
  vertex back to its original edge. The solution clearly has profit $p(T) = w(T')$ and is optimal,
  because a more profitable subgraph with respect to $p$ would also correspond to a
  higher-scoring subgraph with respect to $w$, contradicting the
  optimality of $T'$. \qed
\end{proof}

These reductions directly imply and simplify a number of results for
MWCS\@. For example, it follows from \cite{Feigenbaum:2000jf} and Theorem~\ref{thm:1}
that MWCS is NP-hard and even hard to approximate within a constant
factor. In addition, the results in \cite{Bateni:2011vi} provide a polynomial-time
exact algorithm for MWCS for graphs of bounded treewidth.

%

In \cite{Dittrich:2008hf} we used the close relation to
PCST to develop an exact algorithm for MWCS by running the branch-and-cut approach of Ljubic et al.\
\cite{AlvarezMiranda:2013ew} on the transformed instance. Backes et
al.\ \cite{Backes:2012us} presented a direct integer linear programming
formulation for a variant of MWCS based only on node
variables. {\'A}lvarez-Miranda et al.\ \cite{AlvarezMiranda:2013ew}
recently introduced a stronger formulation based on the concept of
node-separators. 

Here, we introduce an algorithm engineering approach that combines
existing and new results to solve MWCS instances
efficiently in practice to provable optimality. We describe new and
adapted preprocessing rules in Section~\ref{sec:prepr}.
Section~\ref{sec:divide-and-conquer-scheme} is dedicated to an overall
divide-and-conquer scheme, which is based on novel results on
decomposing an instance into its biconnected and triconnected components.
In Section~\ref{sec:branch_and_cut} we describe a branch-and-cut approach using
a new primal heuristic based on an exact dynamic programming algorithm
for trees.  We
demonstrate in
Section~\ref{sec:results} the performance of our approach and the
benefits of preprocessing and the divide-and-conquer scheme.


%


\section{Preprocessing}
\label{sec:prepr}

We describe reduction rules that simplify an instance
of MWCS without losing optimality. We define three classes of
increasingly complex reduction rules and apply them
exhaustively in successive phases of a preprocessing scheme, see Figure~\ref{fig:preproc_scheme}. 

\begin{figure}[hbpt]
  \centering
  \includegraphics[width=\linewidth]{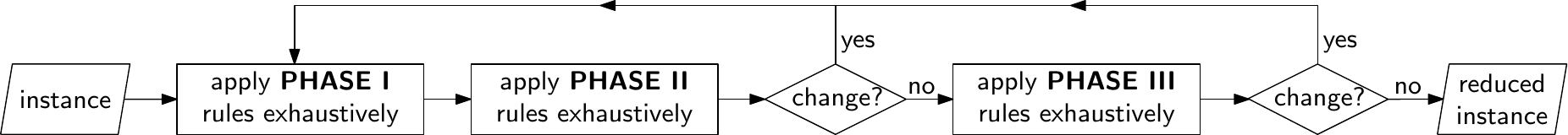}
  \caption{\textbf{Preprocessing scheme.} An MWCS instance passes through three phases of increasingly complex rules that are run exhaustively until no rules apply anymore. The result is a reduced instance. }
  \label{fig:preproc_scheme}
\end{figure}

The rules make use of three operations on node sets: \merge,
\isolate and \remove, see Figure~\ref{fig:operations}. Given a node set
$V'$, \merge$(V')$ combines the nodes in $V'$ into a supernode of
weight $\sum_{v \in V'}w(v)$, which is connected to all neighbors of
nodes in $V'$ outside $V'$. Operation \isolate($V'$) adds a copy of
$V'$ without edges and merges it. Operation \remove($V'$) removes all
nodes in $V'$ from the graph. We keep a mapping from the merged nodes
to sets of original nodes to map solutions of the reduced instance to
solutions of the original instance. These operations will also be used
in our divide-and-conquer scheme, which we will present in
Section~\ref{sec:divide-and-conquer-scheme}. 

\begin{figure}[hbt]
  \centering
  \includegraphics[width=\linewidth]{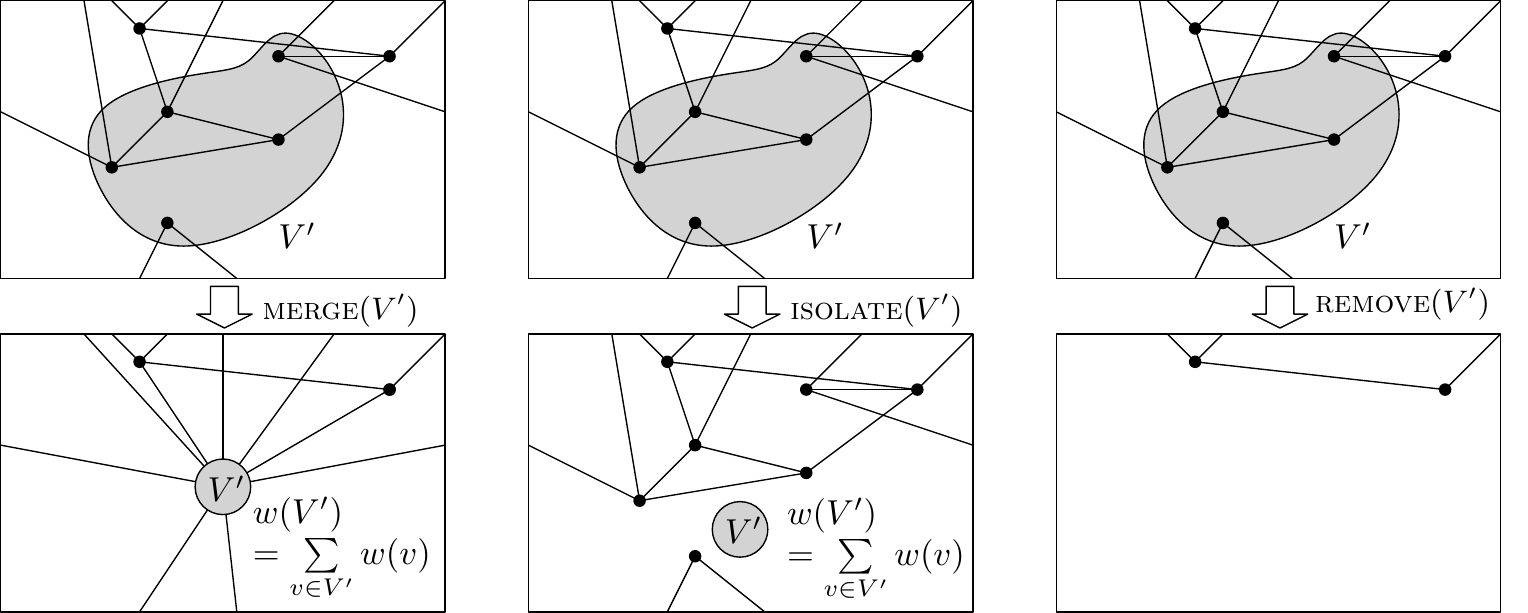}
  \caption{\textbf{Operations} \textsc{Merge}, \textsc{Isolate}, \textsc{Remove}.}
  \label{fig:operations}
\end{figure}

\begin{itemize}
\item \textbf{Phase I rules.} The first phase consists of three
simple rules. 
\begin{enumerate}
\item \emph{Remove isolated negative node rule.} Let $v$ be an isolated vertex with $w(v) < 0$. We can safely
  remove $v$ by calling \remove($\{v\}$), because it will
  never be part of any optimal solution. Identifying all nodes that satisfy the
  condition takes $O(|V|)$ time.
\item \emph{Merge adjacent positive nodes rule.} Let $(u, v)$ be an edge with $w(u) > 0$ and $w(v) > 0$. If one
  vertex will be part of the solution the other one will be as well,
  so we perform \merge($\{u, v\}$). Finding all adjacent positively-weighted
  nodes takes $O(|E|)$ time.
\item \emph{Merge negative chain rule.} Let $P$ be a chain of negative degree 2 vertices. It is safe to
  perform \merge($P$). Either none of the vertices in $P$ will be part
  of an optimal solution or all of them. In the latter case $P$ is used as a
  bridge between positive parts. Identifying all negatively-weighted chains
  takes $O(|E|)$ time.
\end{enumerate}
\item \textbf{Phase II rules.} The second phase consists of one rule. 
  \begin{enumerate}
    \item \emph{Mirrored hubs rule.} Let $u, v \in V$ be two distinct
      negatively-weighted nodes, i.e, $w(u) < 0$ and $w(v) < 0$. Without loss of
      generality assume that $w(u) \leq w(v)$. If $u$ and $v$ are adjacent to
      the same nodes then we can \remove($\{u\}$). The reason is that $v$ will always be
      preferred over $u$ in an optimal solution, because it is adjacent to
      exactly the same nodes as $u$ and costs less. Finding all pairs of
      negatively-weighted mirrored nodes takes $O(\Delta \cdot |V|^2 )$ time where $\Delta$
      is the maximum degree of the graph.
  \end{enumerate}
\item \textbf{Phase III rule.} The last phase consists of the most
  expensive rule. 
  \begin{enumerate}
  \item \emph{Least-cost rule.} This rule is adapted from the
    least-cost test, which was described by Duin and Volgenant
    \cite{Duin:1987wo} for the node-weighted Steiner tree problem. Let $(u, v)$
    and $(v, w)$ be two edges in the graph, and let $v$ have degree 2
    and $w(v) < 0$. We construct a directed graph whose node set is
    $V$ and whose arc set $A$ is obtained by introducing for every edge $(a,b)$ in $G$
    two oppositely directed arcs $ab$ and $ba$. We can \remove($\{v\}$), if the shortest path from
    $u$ to $w$ with respect to lengths $d(ab) := \max\{-w(b), 0\}$ for
    all $ab \in A$ is shorter than $-w(v)$. The reason is that if $u$
    and $w$ were to be in an optimal solution there is a better way
    than using $v$. This rule takes $O(|V'| \cdot (|E| + |V| \log
    |V|))$ time where $V'$ is the set of all negative-weighted nodes having
    degree 2.
  \end{enumerate}
\end{itemize}



\section{Divide-and-Conquer Scheme}
\label{sec:divide-and-conquer-scheme}

We propose a three-layer divide-and-conquer scheme for solving MWCS to
provable optimality. It is based on 
decomposing the input graph into its connected, biconnected and
triconnected components. H\"uffner et al.\ have already considered data
reduction rules based on heuristically found separators of size $k$
for the Balanced Subgraph problem \cite{Hueffner:2010ka}. Here, we
present the first data reduction approach that considers all separators of size 1 and 2 in a rigorous
manner by processing them using the block-cut and SPQR tree data structures. 

In the first layer, we consider the connected components of the input $(G, w)$
one-by-one, see Algorithm~\ref{alg:solve_mwcs}. In the next layer, we
construct a block-cut vertex tree $\TB$ for each connected component $C$.
We process the block leaves $B$ of $\TB$ iteratively. Processing a block $B$
of degree 1 will result in the removal of $B
\setminus \{c\}$, where $c$ is the corresponding cut vertex. In addition, a new
degree 0 node may be introduced. Processing a block $B$
of degree 0 will result in the replacement of $B$ by
a single isolated
node. Therefore, at the end of the loop, the graph $G[C]$ will only
consist of isolated nodes. Among these nodes, the node with maximum weight
corresponds to the maximum weight connected subgraph of $G[C]$. We retain only
this node in the graph, and remove all other nodes in $C$. After processing
all connected components, a similar situation arises in $G$: each component
is an isolated node, and the solution $V^*$ will
correspond to the node that has maximum weight.

\begin{algorithm2e}[btp]
\ForEach{connected component $C$ of $G$}
{
  \preprocess($C$)\;
  let $\TB$ be the block-cut vertex tree of $C$\;
  \While{$T_B$ has block $B$ of degree 0 or 1}
  {
    \processbicomponent($B$)\;
    update $T_B$\;
  }
  $V_C =$ \solveunrooted($C$)\;
  \merge($V_C$); \remove($C \setminus V_c$)\; 
}
$V^*$ $\gets$ \solveunrooted($G$)\;
\Return $V^*$\;
\caption{\textsc{SolveMWCS}($G = (V, E), w$)}
\label{alg:solve_mwcs}
\end{algorithm2e}  

Next, we describe how to process a block $B$. The idea here is
to account for the situation where the final optimal solution $V^*$ contains
parts of $B$, i.e.\ $V^* \cap B \neq \emptyset$. For this to happen, either $V^*$
must be a proper subset of $B$, or a cut node of $B$ must be part of $V^*$. Since $B$
corresponds to a degree 0 or 1 block in $\TB$, it contains at most one cut node
$c$. Let us consider the case where $B$ does have a cut node $c$, as the other
case is straightforwardly resolved by introducing an isolated node. Two subcases can be distinguished: $c \in V^* \cap B$ and
$c \not \in V^* \cap B$. We encode both
cases using the following gadget.
Let $V_1$ be the unrooted maximum-weight connected subgraph of
$G[B]$, and let $V_2$ be the maximum-weight connected subgraph of $G[B]$
rooted at $c$. The corresponding gadget $\Gamma_1$ is obtained by merging the nodes in
$V_2$, and, if $V_1 \neq V_2$, by additionally introducing an isolated vertex
corresponding to $V_1$---see Figure~\ref{fig:scheme_layer1and2}~D and~E. Replacing $B$ by the gadget preserves optimality as
stated in the following lemma.

\begin{lemma}
  Let $B \subseteq V$ be a block in $G = (V, E)$ containing exactly one cut node
  $c$. Let $G' = G[(V \setminus B) \cup \Gamma_1]$ be the graph where 
  $B$ is replaced by gadget $\Gamma_1$. A maximum weight connected subgraph
  of $G'[U^*]$ has the same weight as a maximum weight connected subgraph
  $G[V^*]$, i.e., $w(U^*) = w(V^*)$.
  \label{lem:bi}
\end{lemma}
\begin{proof}
  The gadget $\Gamma_1$ consists of two parts $V_1$ and $V_2$, which correspond
  to the unrooted and $\{c\}$-rooted maximum weight connected subgraph of $G[B]$,
  respectively. By definition $V_1$ and $V_2$ induce connected subgraphs in $G$.
  Therefore the operations \merge($V_2$) and \isolate($V_1$)---resulting in the
  construction of $\Gamma_1$---combined with the optimality of $V^*$ ensure that
  $w(U^*) \leq w(V^*)$.

  We now distinguish two subcases: $V^* \cap B = \emptyset$ and $V^* \cap B \neq
  \emptyset$. Consider the first case. Since the introduction of the gadget only
  concerns nodes in $B$, we have that $w(U^*) \geq w(V^*)$. Hence, $w(U^*) =
  w(V^*)$.

  In the other case, $V^* \cap B \neq \emptyset$, we either have
  that $c \not \in V^* \cap B$ or $c \in V^* \cap B$. 
  If $c \not \in V^* \cap B$ then $V^* \subseteq B$.
  By construction of the gadget, we then have $w(V_1) = w(V^*
  \cap B) = w(V^*)$. Conversely, if $c \in V^* \cap B$ then $w(V_2) = w(V^*
  \cap B)$. Observe that $w(U^* \setminus \Gamma_1) = w(V^* \setminus B)$.
  Therefore $w(U^*) = w(V^*)$. \qed
\end{proof}

\begin{figure}[hbt]
  \centering
  \includegraphics[width=\linewidth]{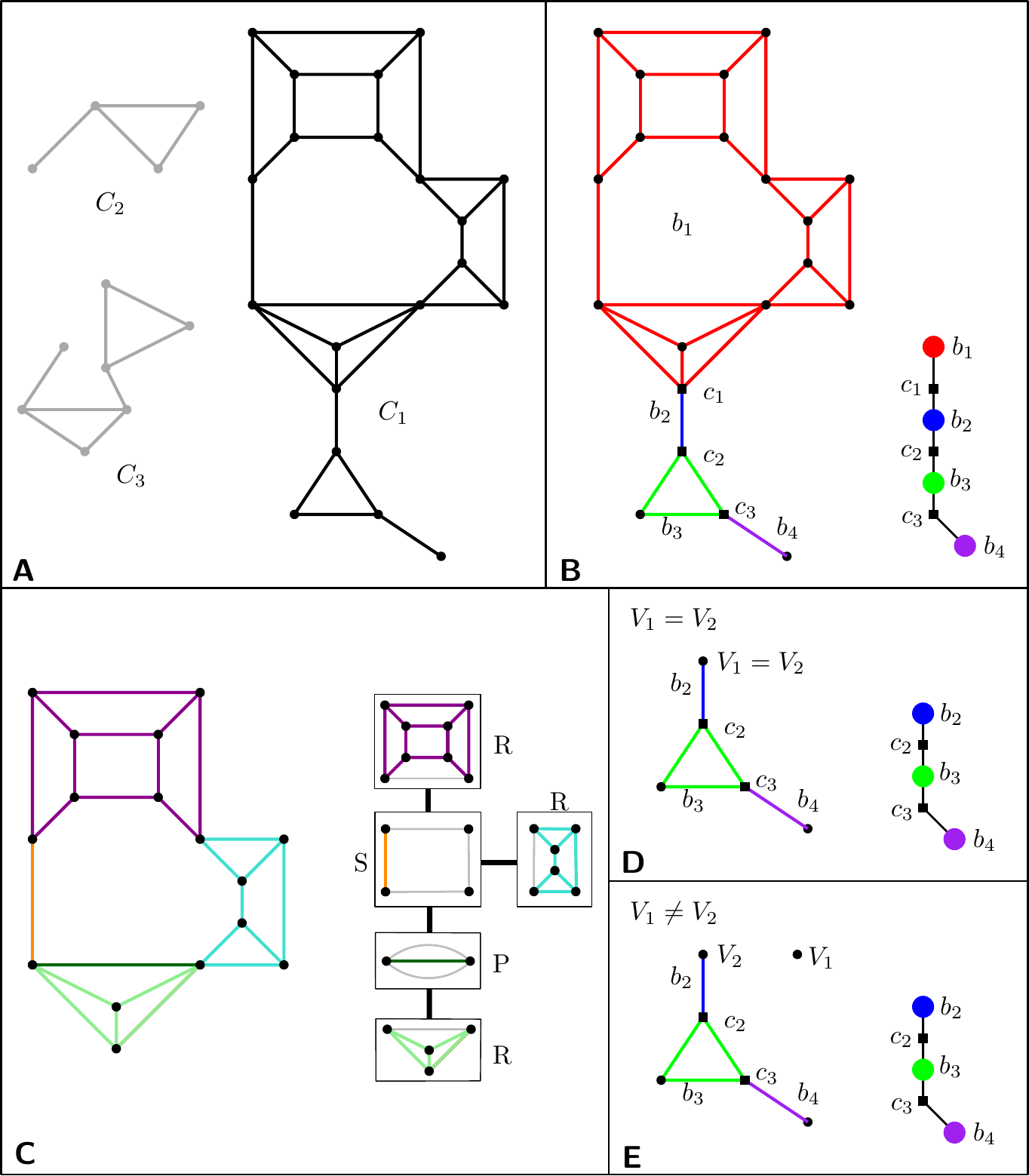}
  \caption{\textbf{The three layers of the divide-and-conquer scheme.}
    \textbf{A:} Three connected components of an MWCS
    instance. \textbf{B:} Biconnected components and the block-cut
    vertex tree of connected component $C_1$. \textbf{C:} Triconnected components
  and the SPQR tree of biconnected component $b_1$. \textbf{D:} Gadget
  $\Gamma_1$ in the first case. \textbf{E:} Gadget $\Gamma_1$ in the
  second case.}
  \label{fig:scheme_layer1and2}
\end{figure}

As an optimization, we preemptively remove a leaf block $B$ if all its nodes $v
\in B \setminus \{c\}$ have nonpositive weights $w(v) \leq 0$. 

In the third layer, we start by constructing an SPQR-tree $\TSPQR$ of $B$. We
then iteratively consider each triconnected component $A$ that does not contain
the cut node $c$ and contains at least three nodes. 
Let $\{u,v\}$ be the cut pair of such a triconnected component $A$. If $A$ consists of only negatively weighted nodes, its only
purpose is to connect $u$ with $v$. To find the cheapest way of doing this, we
construct a directed graph whose node set is $A$ and whose arcs are obtained by
introducing for every edge $(a,b)$ in $G[A]$ two oppositely directed arcs. We
define the cost of an arc $(a,b)$ to be $-w(b)$. The cheapest way of going from
$u$ to $v$ now corresponds to the shortest path from $u$ to $v$ in the directed
graph. Triconnected components that contain positively-weighted nodes are
processed separately and may be replaced by gadgets of smaller size, which we
describe next.

\begin{procedure}[hbtp]
\caption{\textsc{ProcessBicomponent}($B$)}
let $c$ be the corresponding cut node, if applicable\;
\lIf{all $v$ in $B \setminus \{c\}$ have $w(v) \leq 0$}{\remove($B \setminus
  \{c\}$)}
\Else
{
  let $\TSPQR$ be the SPQR tree of $B$\;
  \ForEach{triconnected component $A$ of size $> 3$ not containing $c$}
  {
    let $\{u, v\}$ be the cut pair of $A$\;
    \If{all $v$ in $A$ have $w(v) \leq 0$}
    {
      compute shortest path $P$ from $u$ to $v$\;
      \merge($P \setminus \{u,v\}$); \remove($N \setminus P$)
    }
    \Else{\processtricomponent($A$)}
    \preprocess($B$); update $\TSPQR$\;
  }
  $V_1$ $\gets$ \solveunrooted($B$)\;
  $V_2$ $\gets$ \solverooted($B, \{c\})$\;
  \lIf{$V_1 = V_2$}
  {
    \merge($V_2$); \remove($B \setminus V_2$)
  }
  \lElse
  {
    \isolate($V_1$); \merge($V_2$); \remove($B \setminus V_2$) 
  }
}

\end{procedure}

Let us consider the situation where the final solution $V^*$ contains parts of a
triconnected component $A$ with cut nodes $\{u,v\}$, i.e., $V^* \cap A \neq
\emptyset$. We can distinguish four cases: (i) $u \in V^*$, (ii) $v \in V^*$,
(iii) $\{u,v\} \subseteq V^*$, and (iv) $V^* \subseteq A$. In the following we
introduce a gadget $\Gamma_2$ that encodes all four cases. 
The first three cases correspond to finding a rooted maximum weighted connected subgraph in $G[A]$
with $\{u\}$, $\{v\}$ and $\{u,v\}$ as the root node sets, respectively. 
Let $V_1$, $V_2$, $V_3$ be the solutions sets of the three rooted maximum
weight connected problems from which the respective root nodes have been
removed. The fourth case corresponds to finding an unrooted maximum weight connected subgraph
in $G[A]$ whose solution we denote by $V_4$. To encode the fourth case, we
\isolate set $V_4$. As for the first three cases, we \merge the sets $V_1
\setminus V_2$, $V_2 \setminus V_1$, $V_1 \cap V_2$ and $V_3 \setminus (V_1 \cup
V_2)$ resulting in the nodes $v_1$, $v_2$, $v_3$ and $v_4$, respectively. As some
of these sets may be empty, we need to take care when connecting the gadget.
For instance, if $V_1 \setminus V_2 = \emptyset$ and $V_1 \cap V_2 \neq
\emptyset$ then we need to connect $u$ directly with $v_3$. Also, we ensure that we do not break biconnectivity. For instance, if $V_1
\cap V_2 = \emptyset$ and $V_1 \neq \emptyset$ then we merge $v_1$ and $u$ as to
prevent $u$ from becoming an articulation point. See Figure~\ref{fig:tri_gadget} and the
pseudocode below for more details.

\begin{procedure}[H]
\caption{\textsc{ProcessTriComponent}($A$)}
let $\{u, v\}$ be the cut pair\;
$V_1$ $\gets$ \solverooted($A$, $\{u\}$) $\setminus \{u\}$\;
$V_2$ $\gets$ \solverooted($A$, $\{v\}$) $\setminus \{v\}$\;
$V_3$ $\gets$ \solverooted($A$, $\{u, v\}$) $\setminus \{u, v\}$\;
$V_4$ $\gets$ \solveunrooted($A$)\;
\isolate($V_4$)\;
$\Gamma_2 \gets \{u,v\}$\;
\lIf{$V_1 \setminus V_2 \neq \emptyset$}
{
  $v_1$ $\gets$ \merge($V_1 \setminus V_2$);
  add edge $(u, v_1)$;
  add $v_1$ to $\Gamma_2$
}
\lIf{$V_2 \setminus V_1 \neq \emptyset$}
{
  $v_2$ $\gets$ \merge($V_2 \setminus V_1$);
  add edge $(v, v_2)$;
  add $v_2$ to $\Gamma_2$
}
\If{$V_1 \cap V_2 = \emptyset$}
{
  \lIf{$V_1 \neq \emptyset$}{\merge($\{u, v_1\}$); remove $v_1$ from $\Gamma_2$}
  \lIf{$V_2 \neq \emptyset$}{\merge($\{v, v_2\}$); remove $v_2$ from $\Gamma_2$}
}
\Else
{
  $v_3 \gets$ \merge($V_1 \cap V_2$); add $v_3$ to $\Gamma_2$\;
  \lIf{$V_1 \subseteq V_2$}{add edge $(u, v_3)$ \textbf{else} add edge $(v_1, v_3)$}
  \lIf{$V_2 \subseteq V_1$}{add edge $(v, v_3)$ \textbf{else} add edge $(v_2, v_3)$}
}
\If{$V_3 \setminus (V_1 \cup V_2) \neq \emptyset$}
{
  $v_4$ $\gets$ \merge($V_3 \setminus (V_1 \cup V_2)$)\;
  add $v_4$ to $\Gamma_2$\;
  add edges $(u, v_4), (v, v_4)$
}
\If{$V_1 \cap V_2 = \emptyset$ and $V_3 \setminus (V_1 \cup V_2) = \emptyset$}
{
  \lIf{$V_1 \setminus V_2 \neq \emptyset$}
  {
    add edge $(v_1, v)$
  }
  \lIf{$V_2 \setminus V_1 \neq \emptyset$}
  {
    add edge $(v_2, u)$
  }
}
\remove($A \setminus \Gamma_2$)
\end{procedure}

\begin{figure}[hbt]
  \centering
  \includegraphics[width=\linewidth]{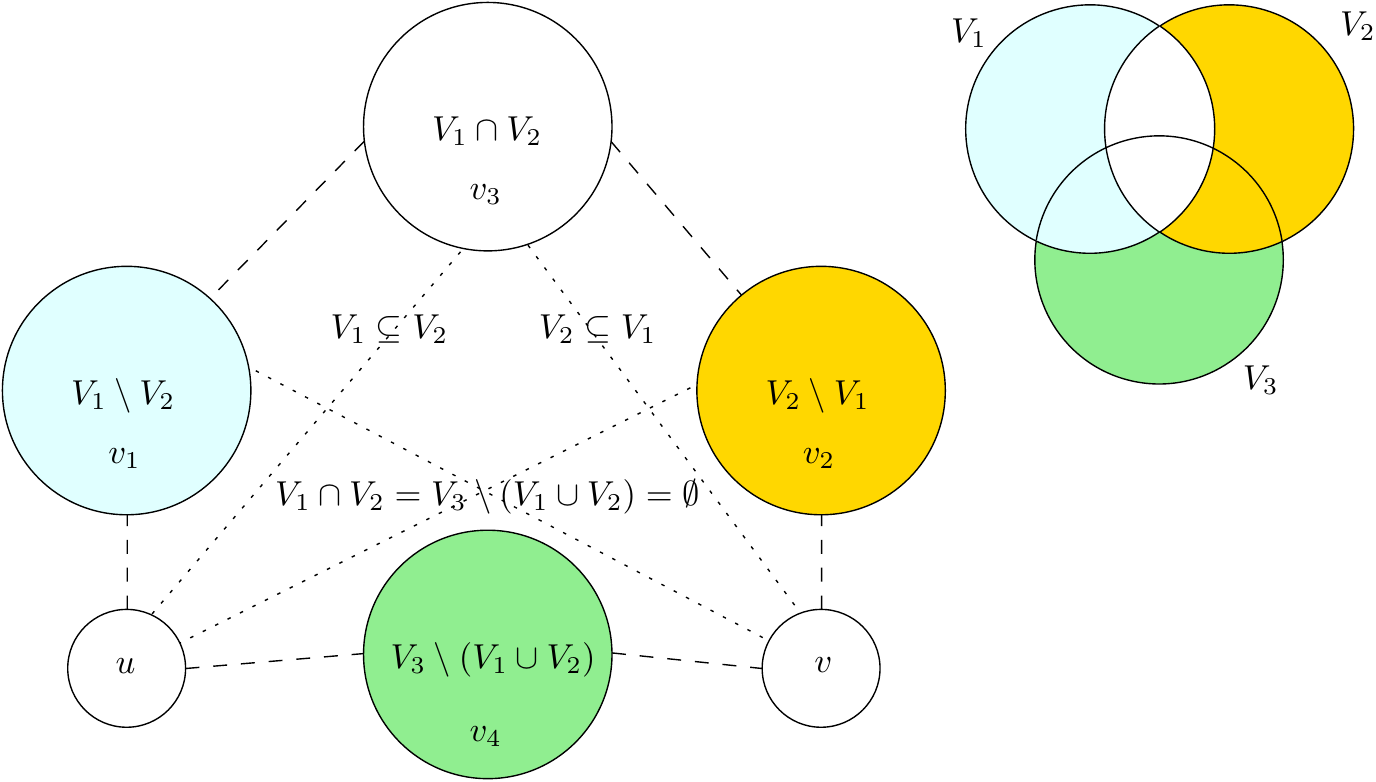}
  \caption{\textbf{Triconnected component gadget.} The gadget $\Gamma_2$
consists of at most four nodes. In case the corresponding set is empty no node
is introduced. The dotted edges are only introduced if the condition on the edge
is met, e.g., there is an edge from $u$ to $v_3$ if $V_1 \subseteq V_2$ and $V_1
\cap V_2 \neq \emptyset$.}
  \label{fig:tri_gadget}
\end{figure}

\begin{lemma}
  Let $A \subseteq V$ be a triconnected component in $G = (V, E)$ not containing
  any cut node of $G$. Let $G' = G[(V \setminus B) \cup \Gamma_2]$ be the graph
  where $A$ is replaced by gadget $\Gamma_2$. 
  A maximum weight connected subgraph
  of $G'[U^*]$ has the same weight as a maximum weight connected subgraph
  $G[V^*]$, i.e., $w(U^*) = w(V^*)$.
  \label{lem:tri}
\end{lemma}

\begin{proof}
  Let $\{u,v\}$ be the cut pair of $A$. The gadget $\Gamma_2$ encodes four node
  sets: $V_1$, $V_2$ and $V_3$ representing the rooted maximum weight connected
  subgraphs of $G[A]$---without their respective root nodes---rooted at $\{u\}$,
  $\{v\}$ and $\{u,v\}$, respectively;
  and $V_4$ representing the unrooted maximum weight connected subgraph of
  $G[A]$. Let $v_1 := \merge(V_1 \setminus V_2)$, $v_2 := \merge(V_2 \setminus
  V_1)$, $v_3 := V_1 \cap V_2$ and $v_4 := V_3 \setminus (V_1 \cup V_2)$---see
  Figure~\ref{fig:tri_gadget}.

  We start by proving $w(U^*) \leq w(V^*)$.
  Since $A$ is a triconnected component, we have that $V_1 \setminus V_2$,
  $V_2 \setminus V_1$, $V_1 \cap V_2$ and $V_3 \setminus (V_1 \cup V_2)$ are
  connected in $G$. In addition, as these node sets are obtained by \merge
  operations only and they are pairwise disjoint, we have that $w(U^*) \leq w(V^*)$.

  We distinguish two cases: $V^* \cap A = \emptyset$ and $V^* \cap A \neq
  \emptyset$. The first case holds, because introduction of the gadget $\Gamma_2$ only concerns
  nodes in $A$. Therefore, $w(U^*) \geq w(V^*)$, which implies $w(U^*) =
  w(V^*)$. The second case, $V^* \cap A \neq \emptyset$, has the following four
  subcases:
  \begin{enumerate}
    \item $u \not \in V^*$ and $v \not \in V^*$;\\
      This implies that $V^*
      \subseteq B$. We then have $w(V_4) = w(V^* \cap A) = w(V^*)$.
    \item $u \in V^*$ and $v \not \in V^*$;\\
      By optimality of $V^*$, we have that
      $w(V_1 \cup \{u\}) = w(u) + w(v_1) + w(v_3) = w(V^* \cap A)$. Since $w(U^*) \leq w(V^*)$, it
      follows that $w(U^*) = w(V^*)$.
    \item $u \not \in V^*$ and $v \in V^*$;\\
      Symmetric to previous subcase.
    \item $u \in V^*$ and $v \in V^*$;\\
      There are two cases: $V_1 \cap V_2 = \emptyset$ or $V_1 \cap V_2 \neq
      \emptyset$. In the first case, we have that $w(V_3 \cup \{u,v\}) = w(u) + w(v) + w(v_1)
      + w(v_2) + w(v_4) = w(V^* \cap A)$.
      In the second case, we have that $w(V_3 \cup \{u,v\}) = w(u) + w(v) + w(v_1) + w(v_2) +
      w(v_3) = w(V^* \cap A)$.
      Since $w(U^*) \leq w(V^*)$, it follows in both cases that $w(U^*) =
      w(V^*)$. \qed
  \end{enumerate}
\end{proof}

Lemmas~\ref{lem:bi} and~\ref{lem:tri} imply the correctness of our
divide-and-conquer scheme.

\begin{theorem}
  Given an instance of MWCS Algorithm~\ref{alg:solve_mwcs} returns an
  optimal solution.
\end{theorem}

%

\section{Branch-and-Cut Algorithm}
\label{sec:branch_and_cut}

To solve the nontrivial instances within our divide-and-conquer
scheme, we use a branch-and-cut approach. We obtain strong upper
bounds from solving the linear programming (LP) relaxation of an integer
linear programming formulation and lower bounds from an
integrated primal heuristics that is guided by the optimal solution of
the LP relaxation. 

\subsection{Integer linear programming formulation}

We use a formulation that only used node variables for
both the unrooted and the rooted MWCS problem. The formulations are
equivalent to the generalized node-separator formulation described in
\cite{AlvarezMiranda:2013ew}. 

\subsubsection{Unrooted.} Variables $\mathbf{x} \in \{0,1\}^V$ encode the
presence of a node in the solution. To encode connectivity in the unrooted case,
we use auxiliary variables $\mathbf{y} \in \{0,1\}^V$ that encode the presence
of the root node. The ILP is as follows.
\begin{alignat}{3}
\label{eq:obj}  \max\: & \sum_{v \in V} w_v x_v \\
\label{eq:sumy} & \sum_{v \in V} y_v = 1 & \\
\label{eq:y}    & y_v \leq x_v & \forall v \in V\\
\label{eq:cut}  & x_v \leq \sum_{u \in \delta(S)} x_u + \sum_{u \in S} y_u
                  \quad\quad & \forall v \in V, \{v\} \subseteq S \subseteq{V}\\
\label{eq:vars_x} & x_v \in \{0,1\} & \forall v \in V\\
\label{eq:vars_y} & y_v \in \{0,1\} & \forall v \in V
\end{alignat}
Constraint~\eqref{eq:sumy} states that there is exactly one root node. A node
can only be the root node if it is present in the solution, which is captured by
constraints~\eqref{eq:y}. Constraints \eqref{eq:cut} state that a node $v$ can
only be present in the solution if for all sets $S$ containing $v$, either the
root node is in $S$, or a node in the set $\delta(S) = \{ v \in V \setminus
S \mid \exists u \in S: (u,v) \in E \}$ is in the solution. In the next
subsection we describe how we separate these constraints.

To strengthen the formulation, we use the following additional cuts.
\begin{alignat}{3}
\label{eq:posy} & y_v = 0 & \forall v \in V, w(v) < 0\\
\label{eq:symmetry} & \sum_{v > u} y_v \leq 1 - x_u & \forall u \in V, w(u) > 0\\
\label{eq:neg}  & x_v \leq x_u & \forall(u,v) \in E, w(u) > 0, w(v) < 0\\
\label{eq:get_in_get_out} & 2 \cdot x_v \leq \sum_{u \in \delta(v)} x_u & \forall v \in V, w(v) < 0\\
\label{eq:cut_easy}     & x_v \leq y_v + \sum_{u \in \delta(v)} x_u \quad\quad & \forall v \in V
\end{alignat}
In \eqref{eq:posy} we require the root node to have a strictly positive weight.
We use symmetry breaking constraints \eqref{eq:symmetry} to force the node with
the smallest index to be the root node. Constraints \eqref{eq:neg} state that a
negatively-weighted node can only be in the solution if all its adjacent
positively-weighted nodes are in the solution. In addition, the presence of a
node with negative weight in the solution implies that at least two of its
neighbors must be in the solution, which is modeled by constraints
\eqref{eq:get_in_get_out}. Constraints~\eqref{eq:cut_easy} are implied by
\eqref{eq:cut} in the case that $|S| = 1$. Adding these constraints results in a
tighter upper bound in the initial node of the branch-and-bound tree.

\subsubsection{Rooted.} The rooted formulation is as follows.
\begin{alignat}{3}
\label{eq:obj_rooted}  \max\: & \sum_{v \in V} w_v x_v \\
\label{eq:roots} & x_r = 1 & \forall r \in R\\
\label{eq:cut_rooted}  & x_v \leq \sum_{u \in \delta(S)} x_u \quad\quad &
\forall r \in R, v \in V \setminus R, \{v\} \subseteq S \subseteq{V \setminus
  \{r\}}\\
\label{eq:vars_x_rooted} & x_v \in \{0,1\} & \forall v \in V
\end{alignat}
Constraints \eqref{eq:roots} enforce the presence of root nodes in the solution.
The cut constraints \eqref{eq:cut_rooted} state that a node $v \in V
\setminus R$ can only be in the solution if for any root $r \in R$ and for all supersets
$S \subseteq V \setminus \{r\}$ of $v$ it holds that a node in the set $\delta(S)$ is in the solution.

We strengthen the formulation using the following cuts.
\begin{alignat}{3}
\label{eq:neg_rooted}  & x_v \leq x_u & \forall(u,v) \in E, w(u) > 0, w(v) < 0\\
\label{eq:cut_easy_rooted}  & x_v \leq \sum_{u \in \delta(v)} x_u \quad\quad &
\forall v \in V \setminus R
\end{alignat}
Constraints \eqref{eq:neg_rooted} are the same as constraints \eqref{eq:neg} for
the unrooted case. Similarly to the unrooted formulation,
constraints~\eqref{eq:cut_easy_rooted} correspond to manually adding cuts
for the case that $|S| = 1$ in \eqref{eq:cut_rooted}.

\subsection{Separation}

\subsubsection{Unrooted.} Similarly to \cite{AlvarezMiranda:2013ew}, the separation problem in the unrooted
formulation corresponds to a minimum cut problem on an auxiliary directed
support graph $D$ defined as follows: each node $v \in V$ corresponds to
an arc $(v_1, v_2)$, and each edge $(u,v) \in E$ corresponds to two arcs
$(u_2, v_1)$ and $(v_2, u_1)$. In addition, an artificial root node $r$ is
introduced as well as arcs $(r,v_1)$ for all $v \in V$. Given a fractional
solution $(\mathbf{\bar{x}},\mathbf{\bar{y}})$, the arc capacities $c$ are set as
follows: $c(r,v_1) = \bar{y}_v$, $c(v_1,v_2) = \bar{x}_v$ and $c(v_2,u_1) = 1$
for all distinct $u,v \in V$. Given a node $v \in V$, we identify violated
constraints by solving a minimum cut problem from $r$ to $v_2$. Let $C$ be a
minimum cut set from $r$ to $v_2$. In case the cut value $c(C)$ is smaller than
$\bar{x}_v$, the cut set will admit a set $S$ and $\delta(S)$ such that
$\bar{x}_v > \bar{x}(\delta(S)) + \bar{y}(S) = c(C)$. We add such violated
constraints to the formulation and resolve again.

\subsubsection{Rooted.} For the rooted formulation the auxiliary graph $D$ is
defined as follows: each node $v \in V \setminus R$ corresponds to
an arc $(v_1, v_2)$, and each edge $(u,v) \in E$ corresponds to two arcs
$(u_2, v_1)$ and $(v_2, u_1)$ if both $u$ and $v$ not in $R$. For each root node
$r \in R$, a single node is introduced in $D$. Edges $(r,v)$ incident to a root node $r
\in R$ where $v \not \in R$ correspond to an arc $(r,v_1)$. We identify violated
constraints by identifying minimum cuts between $r$ and $v_2$ for all $r \in R$
and $v \in V \setminus \{r\}$.


\subsection{Primal heuristic}

As stated in Section~\ref{sec:introduction}, MWCS is solvable in polynomial time
for graphs of bounded treewidth. In fact, for trees R-MWCS is solvable in
linear time by first rooting the tree at a node $r \in R$ and then solving a dynamic
program based on the recurrence:
\begin{equation*}
  M(v) = w(v) + \sum_{u \in \delta^+(v) \setminus R} \max\{M(u), 0\} + \sum_{u
  \in \delta^+(v) \cap R} M(u),
\end{equation*}
where $\delta^+(u)$ are the children of the node $u$.

Our primal heuristic transforms the input graph into a tree by considering the
fractional values $\mathbf{\bar{x}}$ given by the solution of the LP relaxation. We use these values to assign an edge cost
$c(u,v) = 2 - (\bar{x}_u + \bar{x}_v)$ for each edge $(u,v) \in E$. Next, we
compute a minimum-cost spanning tree using Kruskal's
algorithm~\cite{Kruskaljoin:1956kc}. In the unrooted MWCS case, we root the
spanning tree at every positively-weighted node $r$ and assign the solution with
maximum weight to be the primal solution. This leads to running time $O(|V|^2)$. In the R-MWCS case, we only root the spanning tree once at an
arbitrary vertex $r \in R$, resulting in running time $O(|V|)$.

\subsection{Implementation details}

Since CPLEX version 12.3, there is a distinction between the user cut callback
and the lazy constraint callback. The latter is only called for integral
solutions, see Figure~\ref{fig:cplex}. Separation of \eqref{eq:cut}
in the case of integral $(\mathbf{\bar{x}}, \mathbf{\bar{y}})$ can be done by
considering the connected components of the induced subgraph
$G[\mathbf{\bar{x}}]$. Let $r$ be the root node encoded in $\mathbf{\bar{y}}$.
Recall that \eqref{eq:sumy} ensures that there is only one root node.
A connected component $C$ of $G[\mathbf{\bar{x}}]$ that does not contain $r$
corresponds to a violated constraint with $S := C$ and $\delta(S) := \delta(C)$.
Violated constraints for R-MWCS in the case of integrality can be separated
analogously.

\begin{figure}[hbpt]
  \centering
  \includegraphics[width=.6\linewidth]{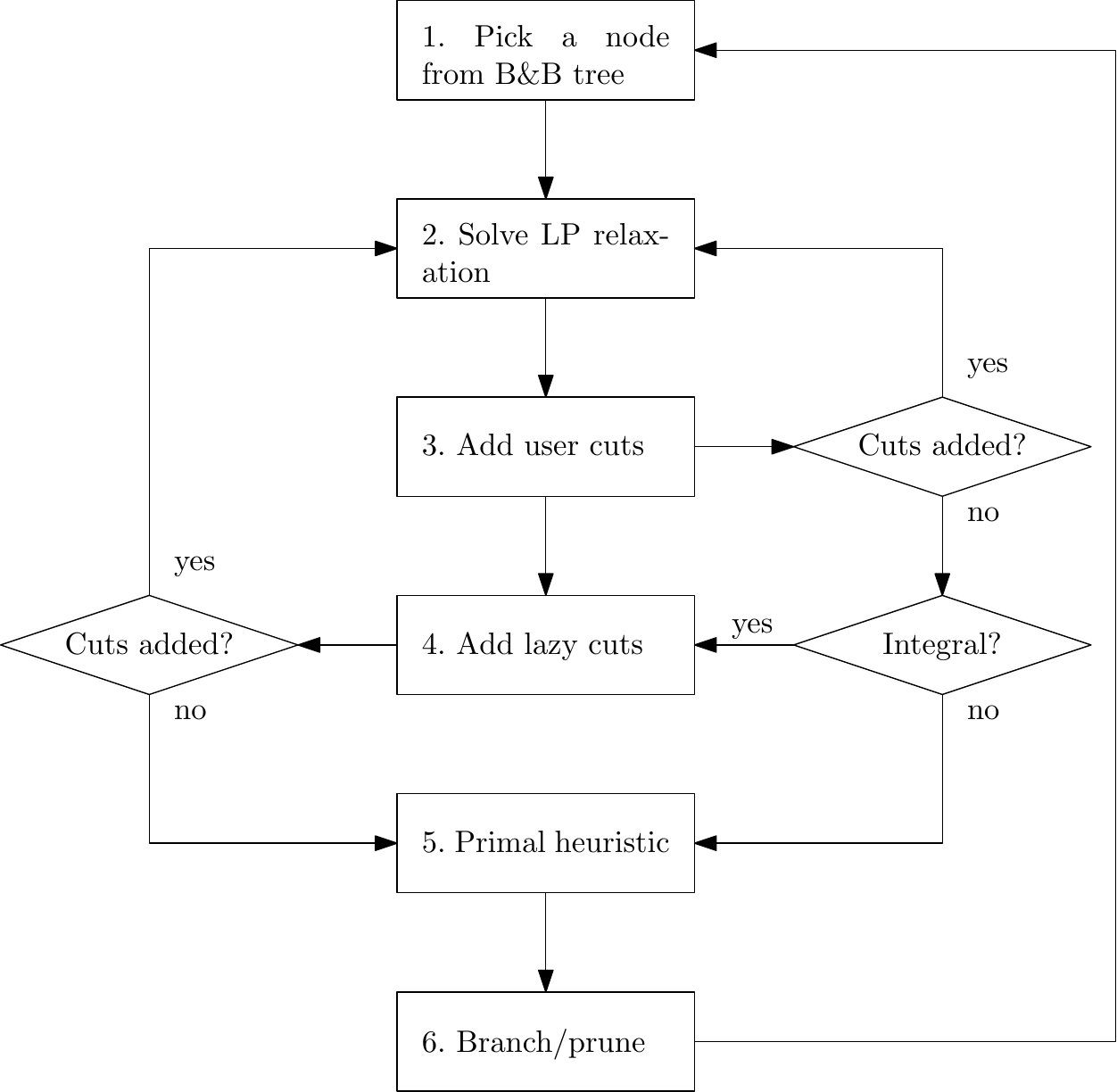}
  \caption{\textbf{CPLEX flow diagram.}}
  \label{fig:cplex}
\end{figure}

As can be seen in Figure~\ref{fig:cplex}, CPLEX calls the user cut callback at
every considered node in the branch-and-bound tree. To prevent spending too much time in the
separation and to allow more time for branching, we choose not to separate violated
constraints at every callback invocation. Instead we make use of a linear
back-off function with an initial waiting period of 1. Upon a successful
attempt, the waiting period is incremented by one, thereby gradually
decreasing the time spent in separating violated constraints.

\section{Results on DIMACS Benchmark}
\label{sec:results}

We implemented our algorithm in \CPP using the LEMON graph library
\cite{Dezso:2011vo}, the OGDF library \cite{Chimani:2011vp} for building the SPQR tree and the CPLEX
v12.6 library for implementing the branch-and-cut approach. Our software tool is called
\heinztwo and is available for download at
\url{http://software.cwi.nl/heinz}. The code of the \heinztwo software is
managed using github and publicly available under the MIT license at \url{https://github.com/ls-cwi/heinz}.

We ran all computational experiments on a 12 core Linux machine with a
2.26 GHz Intel Xeon Processor L5640 and 24~GB of RAM, using 2 threads per instance.
We used all MWCS instances from the 11th DIMACS Implementation
Challenge (\url{http://dimacs11.cs.princeton.edu}). These are the
\ACTMOD set of 8 instances from integrative network analysis in
systems biology and the \JMPALM set of 72 instances, which are based
on the random
Euclidean instances introduced in \cite{Johnson:2000vs}. We also
considered prize-collecting Steiner tree instances from the DIMACS
benchmark, transforming them to MWCS instances using the rule given in
Section~\ref{sec:introduction}. These are \JMP (34
instances), \CRR (80), \PUCNU (18), \isixhundertforty (100), \Hinstances (14),
\HTwoinstances (14) and
\RANDOM (68). In total we ran computational experiments on 408 
instances coming from different applications.

We ran three versions of \heinztwo: (i) A pure branch-and-cut
approach without preprocessing, to establish a baseline, (ii)
preprocessing followed by branch-and-cut, to evaluate the effects of
data reduction and (iii) the divide-and-conquer scheme described in
Section~\ref{sec:divide-and-conquer-scheme}, to evaluate the benefits of
the results described in this paper. To allow for a fair
comparison, we report only results on instances for which all three
methods found feasible solutions. This resulted in 271 instances. 
A full table of results for all these instances is in the appendix.

For each instance we recorded its size in terms of number of nodes and
edges, before and after preprocessing, the best upper and lower bounds
that could be found by each of the three methods within a time limit
of 6 hours wall time, the running time in wall time, as well as the
number of processed biconnected and triconnected components for the
divide-and-conquer scheme.

Figure~\ref{fig:effect_preprocessing} shows the effect of
preprocessing. We can observe that preprocessing is effective,
reducing more than half of the instances to at most 84\% of their original size. Some
instances can even be solved by preprocessing.
\begin{figure}[htpb]
  \centering
  \includegraphics[scale=.5]{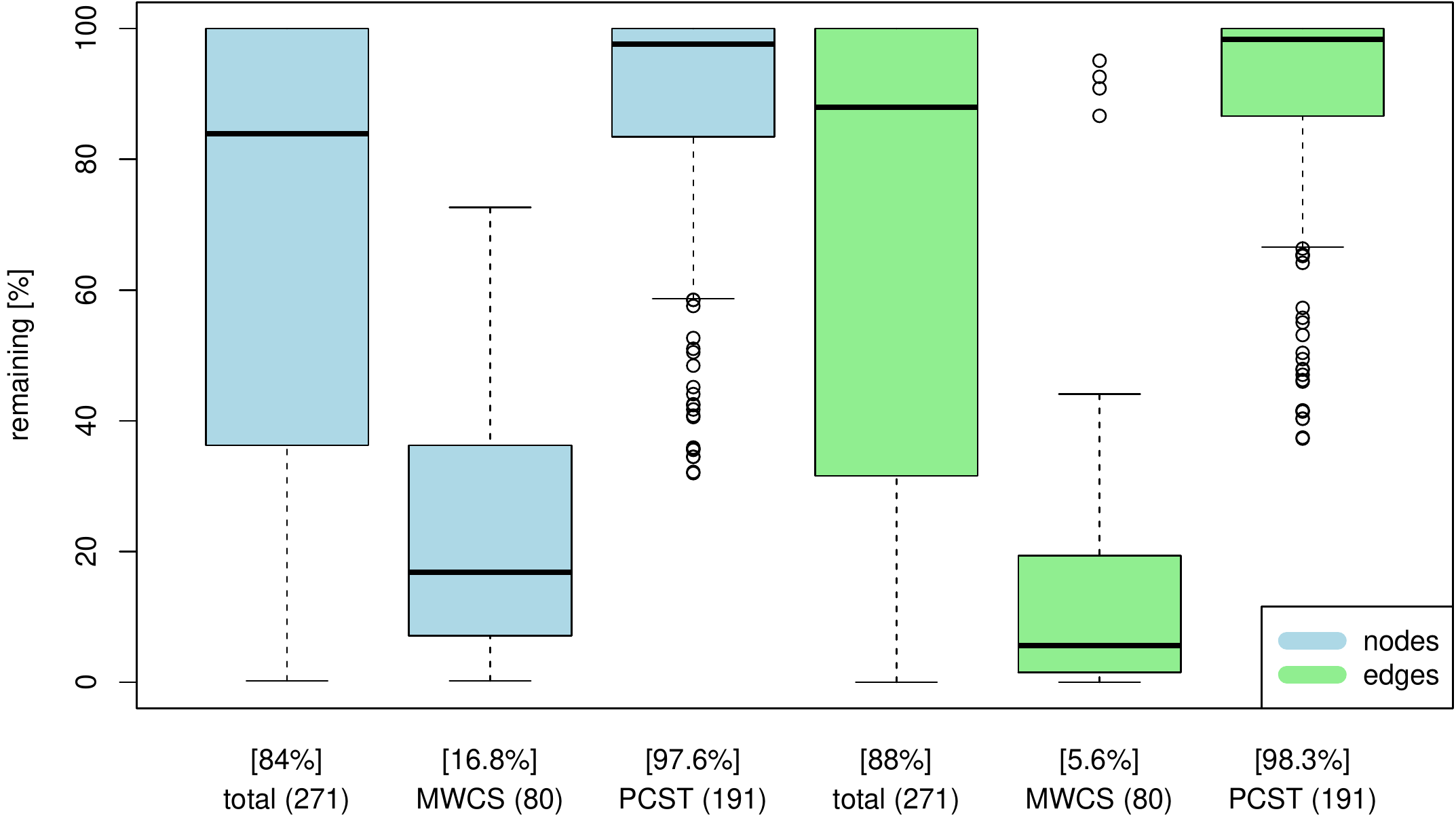}
  \caption{\textbf{Effect of preprocessing.} The boxplots show the
    reduction in number of nodes and edges after preprocessing as a
    fraction of the original value for the 271 instances.
    The median value is shown in between square brackets, and the
    number of instances is in between parentheses. }
  \label{fig:effect_preprocessing}
\end{figure}
Figure~\ref{fig:boxplot_gap} shows the distribution of the optimality
gap for the different version of \heinztwo. It can be seen that while
some instances are hard to solve, both preprocessing and the novel
divide-and-conquer scheme provide significant improvements. Also, it can be seen
that the PCST instances are harder than the MWCS instances for which all
three methods achieve a median gap of 0\% 
\begin{figure}[hbtp]
  \centering
  \includegraphics[width=\textwidth]{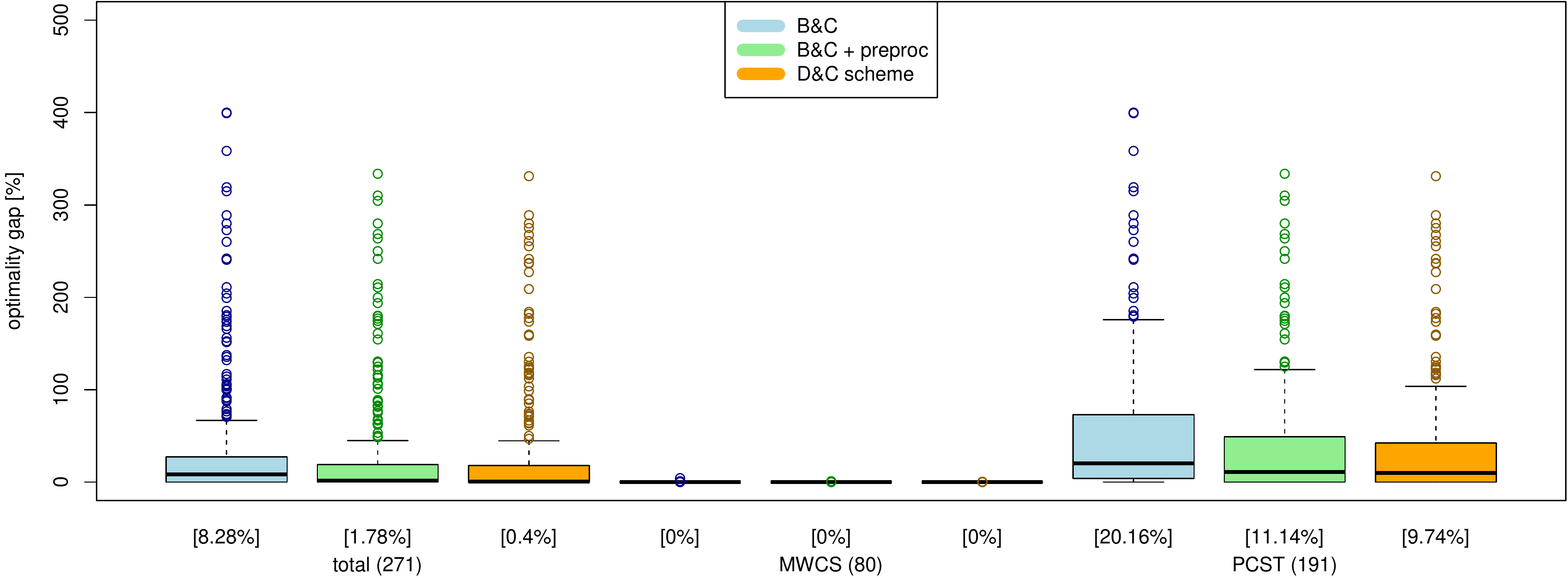}
  \caption{\textbf{Distribution of gaps.} Boxplots of optimality gap for the three different variants of \heinztwo. 
  The median value is shown in between square brackets, and the number
  of instances is in between parentheses.}
  \label{fig:boxplot_gap}
\end{figure}
Figure~\ref{fig:boxplot_time} shows the distribution of the running times of the
instances that were solved to optimality by all three methods. We can see that the divide-and-conquer
scheme (median running time of $0.5$~s) is
faster than the branch-and-cut approach without preprocessing (median running
time of $16.4$~s). On the MWCS instances, the branch-and-cut approach with
processing achieves the same median running time of $0.4$~s as the
divide-and-conquer scheme. For the PCST instances, however, the
divide-and-conquer scheme has the lowest median running time ($3.3$~s).
Moreover, the number of instances that were solved to optimality is the highest
for the divide-and-conquer scheme (134), followed by the branch-and-cut approach
with preprocessing (129) and the branch-and-cut approach without preprocessing
(97).
\begin{figure}[hbtp]
  \centering
  \includegraphics[width=\textwidth]{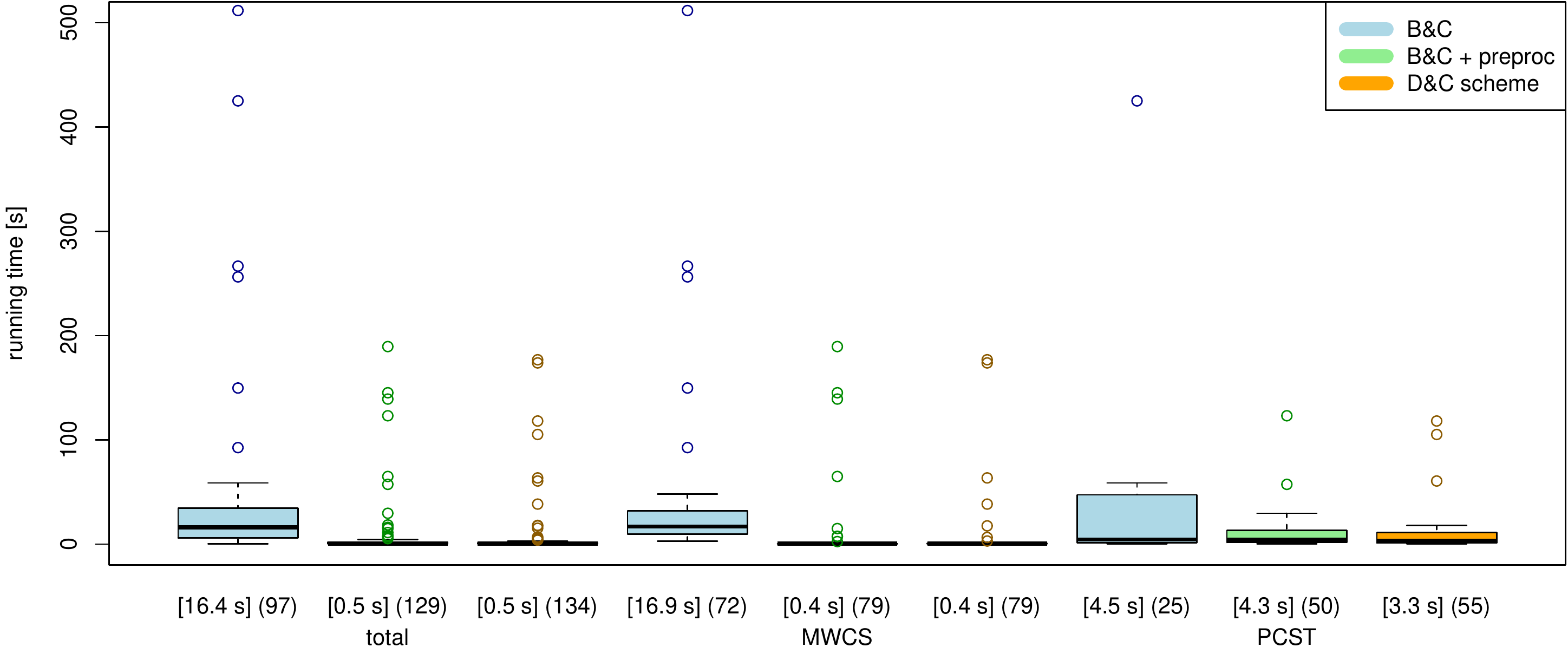}
  \caption{\textbf{Distribution of running times.} Boxplots of running time (s) of the
  instances solved to optimality by all three methods within the time limit.
  The median value is shown in between square brackets, and for each method the total number of instances 
  that it solved to optimality is in between parentheses.}
  \label{fig:boxplot_time}
\end{figure}

\section{Conclusions}
\label{sec:conclusions}

We have presented a divide-and-conquer scheme for solving the
maximum-weight connected subgraph problem to provable optimality. The
scheme combines effective preprocessing with a novel
decomposition approach
that divides an instance into biconnected and triconnected components and
solves the core pieces of an instance using branch-and-cut. We have demonstrated the
performance of our scheme on the benchmark instances of the 11th
DIMACS Implementation Challenge. 

The scheme is modular and allows for the integration of
new preprocessing rules or alternative exact algorithms to solve the
core instances.  We plan, for example, to evaluate a branch-and-cut
approach based on an edge-based ILP formulation, which is similar to
the one we used for the prize-collecting Steiner tree problem in
\cite{lwpkmf:pcst:2006}. Also, we plan to implement an FPT algorithm that can be
plugged into the scheme. The modularity of our approach will make it possible to
perform extensive algorithm engineering studies and to improve upon
the results presented in this paper.

We also want to stress that the divide-and-conquer approach is not
specific to MWCS, but also applicable to other types of Steiner
problems in graphs. Vice versa, techniques that have been proven
useful for related problems may be beneficial for solving MWCS, and we
will evaluate their integration into our scheme.




\small
\paragraph*{Acknowledgments}
We thank the participants of the March 2014 NII Shonan Meeting
\emph{Towards the ground truth: Exact algorithms for bioinformatics
  research} and, in particular, Christian Komusiewicz and Falk H\"uffner
for helpful comments. 

\bibliographystyle{splncs03}
\bibliography{MWCS_DIMACS}
\appendix

\begin{landscape}
\section{Detailed Results}

\small
The following table lists the results of the instances for which all three
methods found feasible solutions. The divide-and-conquer scheme is abbreviated
as `dc', the branch-and-cut approach with preprocessing is `no-dc' and the
branch-and-cut approach without preprocessing is `no-pre'. The time is in
seconds. For results by method `dc', the last three columns correspond to (from
left-to-right) the number of considered blocks, the number of considered
triconnected components with at least one nonnegative node and the number of
considered triconnected components that only contain negative nodes. For
`no-dc', the last three columns correspond to number of nodes after
preprocessing, number of edges after preprocessing and number of components
after preprocessing. For results by `no-pre', the last three columns are empty.

\tiny

\pgfplotstabletypeset[col sep=comma,
                      string type,
                      row sep=\\,
                      begin table=\begin{longtable},
                      end table=\end{longtable},
                      every head row/.append style={before row=\hline,after row=\hline\endhead},
                     ]{results_appendix.csv}
\end{landscape}

\end{document}